\documentclass[runningheads, a4paper]{llncs}

\usepackage[affil-it]{authblk}
\usepackage[english]{babel}
\usepackage{amsmath}
\usepackage{amssymb}
\usepackage{mathrsfs}
\usepackage{delarray}
\usepackage{tabularx}
\usepackage[noend]{algpseudocode}
\usepackage{environ}
\usepackage{authblk}
\usepackage{stmaryrd}
\usepackage{calc}
\usepackage{algorithm}
\usepackage{algpseudocode}

\usepackage{graphicx}

\usepackage{xcolor}

\usepackage{tikz}
\usetikzlibrary{arrows, automata, positioning, arrows.meta, calc}
\let\implies\Rightarrow
\newcommand{\B}{\mathbb{B}}

\newcommand{\NP}{$\text{NP}$}
\newcommand{\coNP}{$\text{coNP}$}

\newcommand{\parspacing}{}

\algblock{Input}{EndInput}
\algnotext{EndInput}
\algblock{Output}{EndOutput}
\algnotext{EndOutput}
\newcommand{\Desc}[2]{\State \makebox[4em][l]{#1}#2}

{\itshape}{\upshape}

\title{Reducing block-sequential automata networks to \emph{smaller}
parallel networks with isomorphic limit dynamics}
\date{}

\begin{document}

\title{Turning block-sequential automata networks into smaller
parallel networks with isomorphic limit dynamics}

\author{Pacôme Perrotin\inst{1,2}
  \and
  Sylvain Sen{\'e}\inst{1,2}
}

\titlerunning{Turning block-sequential ANs into smaller parallel networks}
\authorrunning{P. Perrotin, S. Sené}

\institute{
  Universit\'e publique, Marseille, France
  \and
  Aix-Marseille Univ., CNRS, LIS, Marseille, France
}

\maketitle

\begin{abstract}
We state an algorithm that, given an automata network and a block-sequential
update schedule, produces an automata network of the same size or smaller
with the same limit dynamics under the parallel update schedule.
%This algorithm
%allows the direct characterization of disjunctive double cycles in the
%block-sequential update schedule.
Then, we focus on the family of automata cycles which share a unique path of
automata, called tangential cycles, and show that a restriction of our algorithm
allows to reduce any instance of these networks under a block-sequential update
schedule
into a smaller parallel network of the
family and to
characterize the number of reductions operated while conserving their limit
dynamics. We also show
that any tangential cycles reduced by our main algorithm are transformed into a 
network whose size is that of the largest cycle of the initial network. We end
by showing that the restricted algorithm allows the direct characterization of
block-sequential double cycles as parallel ones.
\end{abstract}

\section{Introduction}

Automata networks are classically used to model gene regulatory networks
\cite{J-Kauffman1969,J-Thomas1973}
\cite{J-Mendoza1998,J-Davidich2008,J-Demongeot2010}.
In these applications the dynamics of automata networks help to understand how the biological 
systems might evolve. 
As such, there is motivation in improving our computation and characterization 
of automata networks dynamics. 
This problem is a difficult one to approach considering the vast diversity of 
network structures, local functions and update schedules that are studied.
Rather than considering the problem in general, we look for families or 
properties which allow for simpler dynamics that we might be able to 
characterize~\cite{J-Demongeot2012,Gao2018}.\medskip

We are interested in studying the limit dynamics of automata networks, that is,
the limit cycles and fixed points that they adopt over time,
notably since these asymptotic behaviors of the underlying dynamical
systems may correspond to real biological phenomenologies such as the genetic
expression patterns of cellular types, tissues, or paces.
More precisely,
we are less
interested in the possible configurations themselves than in the information
that is being transfered and computed in networks over time.
As such, given families of networks, one of our
objectives is to count the fixed points and limit cycles
they possess.\medskip

In this paper, we provide an algorithm that, given an automata network and a
block-sequential
update schedule, produces an automata network of the same size or smaller
with isomorphic limit dynamics under the parallel update schedule.
After definitions in Section~\ref{sec-def}, this algorithm is detailed in
Section~\ref{sec-algo}.
%It is then applied in Section~\ref{sec-cycles} to characterize
%the limit dynamics of tangential cycles under any block-sequential
%update schedule, which is built upon a previous characterization of these
%dynamics under the parallel update schedule~\cite{T-Noual2012,dns10,n12b}.
In Section~\ref{sec-cycles},
the feasibility of the algorithm on \emph{Tangential Cycles} (TC)
is studied, a TC being a set of cycles that intersect on a shared path of
automata.\smallskip
%The demonstrations of all results are available in the appendix.

\emph{Why focusing on TCs?}
Cycles are fundamental retroactive patterns that
are necessary to observe complex dynamics~\cite{Robert1980}.
They are present in many biological regulation networks~\cite{J-Thomas1981}
and are perfectly understood in isolation~\cite{dns10,T-Noual2012}.
In theory, cycles generate an exponential amount of limit cycles,
which is incoherent with the observed behavior of biological systems~\cite{J-Kauffman1969}.
The only way to reduce the amount of limit cycles
is to constrain the degrees of freedom induced by isolated cycles,
which can only be done by intersecting cycles from the purely
structural standpoint. This leads us naturally to TCs, as a simple
intersection case.
Double cycles (intersections of two isolated cycles)
in particular are the largest family of intersecting cycles
for which a complete characterization exists~\cite{T-Noual2012,BC-Demongeot2022};
the present paper generalizes this result to block-sequential update schedules.
Moreover, from the biological standpoint, double cycles are also observed
in biological regulation networks, in which they seem to serve as inhibitors
of their limit behavior~\cite{Demongeot2011}.

\section{Definitions}
\label{sec-def}

Let $\Sigma$ be a finite alphabet. We denote by $\Sigma^n$ the set of all words
of size $n$ over the alphabet $\Sigma$, such that for all $1 \leq i \leq n$ and
$x \in \Sigma^n$, $x_i$ is the $i$th letter of that word. 
An \emph{automata network (AN)} is a function $F:\Sigma^n \to \Sigma^n$, where $n$
is the size of the network. A configuration of $F$ is a word over $\Sigma^n$.
The global function $F$ can be divided into
functions that are local to each automaton: $\forall k, f_k : \Sigma^n \to \Sigma$,
and the global function can be redefined as the parallel application of
every local function: $\forall 1 \leq i \leq n, F(x)_i = f_i(x)$.
For convenience, the set of automata $\{1, \ldots, n\}$ is denoted by $S$,
and will sometimes be considered as a set of letters rather than numbers.
For questions of complexity, we consider that
\emph{local functions are always encoded as circuits.}
\parspacing

For $(i, j)$ any pair of automata, $i$ is said to \emph{influence} $j$ if
and only if
%changing the value of $i$ in a configuration $x$ can change the value of $f_j(x)$.
there exists a configuration $x \in \Sigma^n$ in which there exists a state
change of $i$ that changes the state of $f_j(x)$.
More formally, $i$ influences $j$ if and only if there exists
$x, x' \in \Sigma^n$ such that $\forall k, x_k = x'_k \Leftrightarrow k \neq i$
and $f_j(x) \neq f_j(x')$. \parspacing

It is common to represent an automata network $F$ as the digraph with
its automata as
nodes so that $(i, j)$ is an edge if and only if $i$ influences $j$.
This digraph
is called the \emph{interaction digraph} and is denoted by $G_I(F) = (S, E)$,
with $E$ the set of edges.
The automata network described in
Example~\ref{example-AN} is illustrated as an interaction digraph in Figure~\ref{fig-AN}.

\begin{example}
\label{example-AN}
Let $F:\B^3\to\B^3$ be an AN with local functions
\begin{align*}
f_a(x) &= \neg x_b \vee x_c \\
f_b(x) &= x_a \\
f_c(x) &= \neg x_b
\end{align*}

\end{example}

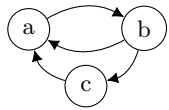
\begin{figure}[t!]
\centering
\begin{tikzpicture}  [-{Latex[length=1.5mm, width=1.5mm]}, node distance=.5cm]
	\node[state, minimum size=.55cm] (A) {a};
	\node[state, minimum size=.55cm] (C) [below right =of A] {c};
	\node[state, minimum size=.55cm] (B) [above right =of C] {b};

	\path
			(A) edge [bend left] (B)
			(B) edge [bend left] (C)
			(C) edge [bend left] (A)
			(B) edge [bend left] (A);

\end{tikzpicture}
\caption{Interaction digraph of the AN detailed in Example~\ref{example-AN}.}
\label{fig-AN}
\end{figure}

An \emph{update schedule} is an infinite sequence of non-empty subsets of $S$,
called blocks.
Such a sequence describes in which order the local functions are to be applied
to update the network, and there are uncountably infinitely many of them.
A \emph{periodic update schedule} is an infinite periodic sequence of non-empty
subsets of $S$, which we directly define by its period.
The application of an update schedule on a configuration of
a network is the parallel application of the local functions of the subsets
in the sequence, each subset being applied one after the other. \parspacing

%An \emph{Update Schedule} is a sequence of non-empty subsets of
%$S$. Such a sequence describes in which order the local functions
%should be applied to update the network. The application this update schedule on a
%configuration of the network is the sequential application of the local functions
%of the subsets in the sequence, each subset being applied in parallel.

For example, the sequence
$\pi = ( S )$ is the parallel update schedule. It is periodic,
and its application
on a configuration is undistinguishable from the application of $F$.
The sequence $(\{1\}, \ldots, \{n\})$ is also a periodic update schedule,
and implies the application of every local
function in order, one at a time. \parspacing

Formally, the application of a periodic update schedule $\Delta$ to a
configuration $x \in \Sigma^n$ is denoted by the function $F_\Delta$, and is defined
as the composition of the
applications of
%each element of $\Delta$, in order.
the local functions in the order specified by $\Delta$.
For any subset $X \subseteq S$, updating $X$ into $x$ is denoted by $F_X(x)$ and
is defined as

\begin{equation*}
	\forall i \in S,\ F_X(x)_i = \begin{array}\{{ll}.
  		f_i(x) & \text{if } i \in X \\
		x_i & \text{otherwise}
	\end{array}\text{.}
\end{equation*}

Example~\ref{example-MMJ} provides an example of the execution of the network
detailed in Example~\ref{example-AN} under some non-trivial update schedule.

\begin{example}
\label{example-MMJ}
Let $\Delta = (\{b, c\}, \{a\}, \{a, b\})$ be a periodic update schedule,
and let $x = 000$ be an initial
configuration. For $F$ the AN detailed in Example~\ref{example-AN}, we have that:
\begin{align*}
  F_{\Delta}(000) &= (F_{\{a, b\}} \circ F_{\{a\}} \circ F_{\{b, c\}}) (000) \\
                  &= (F_{\{a, b\}} \circ F_{\{a\}}) (001) \\
                  &= F_{\{a, b\}} (101) = 111.
\end{align*}

%which means that $F_\Delta(x) = 111$.

\end{example}

A \emph{block-sequential update schedule} is a periodic update schedule where
all the
subsets in a period form a partition of $S$; that is, every automaton
is updated exactly once in the sequence.
If every subset in the sequence is of cardinality $1$, the update schedule is
said to be sequential.
For any AN with automata $S$, both the parallel update schedule and the
$|S|!$ different sequential update schedules are block-sequential.
Block-sequential update schedules are \emph{fair} update schedules, in
the sense that applying it updates
each automaton the same amount of times.\parspacing

The application of a
block-sequential update schedule on an AN
can be otherwise represented as an update digraph, introduced
in~\cite{salinas2008estudio,J-Aracena2009}.
For $F$ an AN and $\Delta$ a block-sequential
update schedule, the \emph{update digraph} of $F_\Delta$, denoted by $G_U(F_\Delta)$,
is an annotation of the network's interaction digraph, where any edge
$(u, v)$ is annotated with $<$ if $u$ is updated strictly before $v$ in
$\Delta$, and with $\geqslant$ otherwise. An update digraph of the AN
detailed in Example~\ref{example-AN} is illustrated in Figure~\ref{fig-UD}. \parspacing

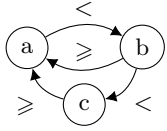
\begin{figure}[t!]
\centering
\begin{tikzpicture}  [-{Latex[length=1.5mm, width=1.5mm]}, node distance=.5cm]
	\node[state, inner sep=3pt, minimum size=.55cm] (A) {a};
	\node[state, inner sep=3pt, minimum size=.55cm] (C) [below right =of A] {c};
	\node[state, inner sep=3pt, minimum size=.55cm] (B) [above right =of C] {b};

	\path
			(A) edge [bend left] node [above] {$<$} (B)
			(B) edge [bend left] node [below right] {$<$} (C)
			(C) edge [bend left] node [below left] {$\geqslant$} (A)
			(B) edge [bend left] node [above] {$\geqslant$} (A);

\end{tikzpicture}
\caption{Update digraph of the AN detailed in Example~\ref{example-AN},
for $\Delta = (\{a\}, \{b\}, \{c\})$.}
\label{fig-UD}
\end{figure}

Given an automata network $F$ and a periodic update schedule $\Delta$,
we define the
\emph{dynamics} of $F_\Delta$ as the digraph
with all configurations $x \in \Sigma^n$
as nodes, so that $(x, y)$ is an edge of the dynamics if and only if
$F_\Delta(x) = y$.
We call \emph{limit cycle of length $k$}
any sequence of unique
configurations $(x_1, x_2, \ldots, x_k)$ such that $F_\Delta(x_i) = x_{i + 1}$
for all $1 \leq i < k$, and $F_\Delta(x_k) = x_1$. A limit cycle of length
one is called a \emph{fixed point}.
The \emph{limit dynamics} of $F_\Delta$ is the subgraph
which contains only the limit cycles and the fixed points of the dynamics.
The limit dynamics of the network defined in Example~\ref{example-AN} are
emphasized in Figure~\ref{fig-DYN}. \parspacing

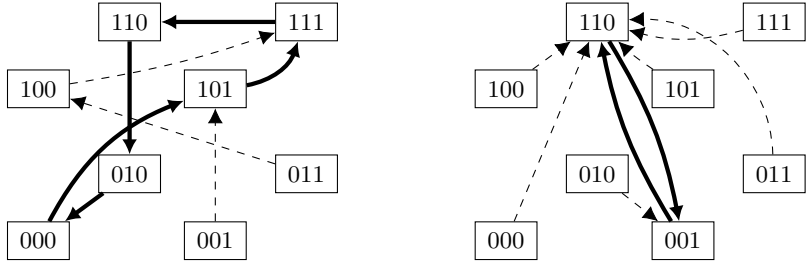
\begin{figure}[t!]
\centering
\tikzset{configuration/.style = {state, rectangle, minimum height=0.5cm}}
\begin{minipage}{.495\textwidth}
\centering
\begin{tikzpicture} [node distance=1.5cm]
	\node[configuration] (000) {000};
	\node[configuration] (001) [right =of 000] {001};
	\node[configuration] (100) [above =of 000] {100};
	\node[configuration] (010) [above right =15pt of 000] {010};
	\node[configuration] (101) [right =of 100] {101};
	\node[configuration] (011) [right =of 010] {011};
	\node[configuration] (110) [above =of 010] {110};
	\node[configuration] (111) [right =of 110] {111};

	\path[-{Latex[length=2mm, width=2mm]}]
			(000) edge [ultra thick, bend left = 20] (101)
			(001) edge [dashed] (101)
			(010) edge [ultra thick] (000)
			(011) edge [dashed] (100)
			(100) edge [dashed, bend right = 5] (111)
			(101) edge [ultra thick, bend right = 30] (111)
			(110) edge [ultra thick] (010)
			(111) edge [ultra thick] (110)

            ;
\end{tikzpicture}
\end{minipage}
\begin{minipage}{.495\textwidth}
\centering
\begin{tikzpicture} [node distance=1.5cm]
	\node[configuration] (000) {000};
	\node[configuration] (001) [right =of 000] {001};
	\node[configuration] (100) [above =of 000] {100};
	\node[configuration] (010) [above right =15pt of 000] {010};
	\node[configuration] (101) [right =of 100] {101};
	\node[configuration] (011) [right =of 010] {011};
	\node[configuration] (110) [above =of 010] {110};
	\node[configuration] (111) [right =of 110] {111};

	\path[-{Latex[length=2mm, width=2mm]}]
			(000) edge [dashed] (110)
			(001) edge [ultra thick, bend left = 10] (110)
			(010) edge [dashed] (001)
			(011) edge [dashed, bend right = 45] (110)
			(100) edge [dashed] (110)
			(101) edge [dashed, bend left = 5] (110)
			(110) edge [ultra thick, bend left = 10] (001)
			(111) edge [dashed, bend left = 15] (110)

            ;
\end{tikzpicture}
\end{minipage}
\caption{Two dynamics of the AN $F$ detailed in Example~\ref{example-AN}.
On the left, the dynamics of $F$ under the parallel update schedule.
On the right, the dynamics of $F$ under the update schedule
$\Delta = (\{a\}, \{b\}, \{c\})$.
The limit dynamics are depicted with bold arrows.}
\label{fig-DYN}
\end{figure}

Since the dynamics of a network is a graph that is exponential in size relative to
the number
of its automaton, naively computing the limit dynamics of a family of network
is a computationally expensive process.

\section{The algorithm}
\label{sec-algo}

In this section, we look at an algorithm that can turn any automata network $F$
with a block-sequential update schedule $\Delta$ into another automata
network $F'$, such that the limit dynamics of $F_\Delta$ stays isomorphic to
the limit dynamics of $F'$ under the parallel update schedule $\pi$.
Furthermore, the size of $F'$ will always be the size of $F$, \emph{or less}. \parspacing

This algorithm is built from two parts: first, we parallelize the network
thanks to a known algorithm in the folklore of automata networks. Second, we remove
automata from the networks based on redundancies created in the first step. \parspacing
%The conditions that allow this reduction in size are precised below.

First, let us state the usual algorithm that, given an automata network $F$ and
a block-sequential update schedule $\Delta$,
provides a new automata network $F'$ defined on the
same set of automata,
such that $F_\Delta$ and $F'_\pi$ have the same exact dynamics. \parspacing

\vspace{-.4cm}
\begin{algorithm}
\caption{Parallelization algorithm of $F_\Delta$}
\begin{algorithmic}
  \Input
  \Desc{$F$}{local functions of a network over $S$, encoded as circuits}
  \Desc{$\Delta$}{block-sequential update schedule over $S$}
%  \Desc{$G_U(F_\Delta)$}{update digraph of $F_\Delta$}
  \EndInput
  \Output
  \Desc{$F$}{local functions of a parallel network over $S$, encoded as circuits}
  \EndOutput \parspacing
%State $G_U(F_\Delta) \gets$ the update digraph of $F, \Delta$
%\For{$(u, <, v) \in E(G_U(F_\Delta))$}
\For{$(u, v)$ such that $u$ precedes $v$ in $\Delta$}
  \State apply the substitution $x_u \mapsto \theta_u$ in $f_v$ \Comment{$\theta$ is a temporary symbol}
\EndFor
\State let $X \gets S$
\While{$|X| > 0$}
  \State let $s \in X$ such that $f_s$ contains no $\theta$ symbol
  \State $X \gets X \setminus \{s\}$
  \For{$s' \in X$}
    \If{$f_{s'}$ contains $\theta_s$}
      \State apply the substitution $\theta_s \mapsto f_s$ in $f_{s'}$
    \EndIf
  \EndFor
\EndWhile
\State return $F$
\end{algorithmic}
\label{algo-partiel}
\end{algorithm}

Algorithm~\ref{algo-partiel} proceeds with two waves of subsitutions. First,
for every $<$-edge $(u, <, v)$, the influencing automaton $u$ is replaced in
the local function of $v$ by a token symbol $\theta_u$.
All of these token symbols are then replaced by the corresponding local
functions (in this case, $f_u$) in the correct order: that is, no function
is ever used in a substitution if it contains a token character. This way,
even if the network contains a complex tree of $<$-edges, the
substitutions will be applied in the correct order. \parspacing
It holds that this algorithm always returns,
and runs in polynomial time.
%The update digraph of $F_\Delta$ is considered
%given as part of the input.

\begin{property}
Algorithm~\ref{algo-partiel} always returns, and does so in polynomial time.
\end{property}
\begin{proof}
Let us denote by \emph{$<$-graph} the subset of the graph $G_U(F_\Delta)$
where only the $<$-edges have been preserved. The $<$-graph of $F_\Delta$ is
always a tree (or multiple disconnected trees): if this wasn't the case, there
would be a cycle of $<$-edges in $G_U(F_\Delta)$, which would mean a cycle
of automata that are all updated strictly before their out-neighbor,
which is impossible. \smallskip

Algorithm~\ref{algo-partiel} will place a $\theta$ symbol for every edge in the
$<$-graph. In the second loop, the selected $s$ is always a leaf of one of the
trees contained in the $<$-graph. The applied substitution removes that leaf
from the $<$-graph. By the structure of a tree, all the $<$-edges will be
removed and the algorithm terminates. \smallskip

To see that this algorithm can be performed in polynomial time, consider that
all of the local functions are encoded as circuits. As such, it is enough to
prepare a copy of each local function into one large circuit, on which every
substitution will be performed.
Any substitution $x_u \mapsto \theta_u$ is performed by renaming the
corresponding input gate. Any substitution $\theta_s \mapsto f_s(x)$ is
performed by replacing the input gate which corresponds to $\theta_s$ by
a connection to the output gate of the circuit that computes the local function
$f_s$.
These substitutions are performed for every $<$-edge in the update
digraph of $F_\Delta$, which can be done by doing one substitution
for every pair in the partial order provided by $\Delta$, which is never
more than $n^2$.
%, which is part of the inputs.
The resulting circuit is then duplicated for every automaton in the output
network, which leads to a total size of no more than
$k^2$, for $k$ the size of the input. \qed
\end{proof}

\begin{remark}
This algorithm is not polynomial if the local functions are encoded as
formul\ae, which is a detail often overlooked in the literature where this
parallelization algorithm is always assumed to be polynomial.
\end{remark}

\begin{theorem}
For any $F_\Delta$
Algorithm~\ref{algo-partiel} returns a network $F'$ such that the dynamics
of $F_\Delta$ is equal to that of $F'_\pi$.
\label{th-partiel}
\end{theorem}
\begin{proof}
Let us consider some configuration $x\in\Sigma^n$, and let us compute its
image $x'$ in both systems.
Let us consider the initial block $X_0$ in $\Delta$. For any automaton in
$X_0$, its local function is untouched in $F'$, and thus
$F_\Delta(x)|_{X_0} = F'(x)|_{X_0}$. Suppose
that $F_\Delta(x)|_{X_0 \cup \ldots \cup X_k} = F'(x)|_{X_0 \cup \ldots \cup X_k}$
for some $k$,
let us prove that is true when including the next block $X_{k+1}$. \smallskip

Let $v \in X_{k+1}$. By the nature of updates in $\Delta$, $f_v$ will be
updated using the values in $F_\Delta(x)$
for any $x_u$ such that $u \in X_0 \cup \ldots \cup X_k$,
and in $x$ otherwise.
In $F'$, in the local function $f'_v$ and for any
$u \in X_0 \cup \ldots \cup X_k$ that influences $v$, a substitution has replaced
$x_u$ by $f'_u(x)$, which implies that the value of $v$ will be updated using
a value of $u$ in $F'(x)$. Pulling this together, we obtain that
$f_v(x) = f'_v(x)$
and $F_\Delta(x)|_{X_0 \cup \ldots \cup X_{k+1}} = F'(x)|_{X_0 \cup \ldots \cup X_{k+1}}$,
and the recurrence yields $F_\Delta(x) = F'(x)$ for any $x$.
\qed

\end{proof}

\vspace{-.4cm}
\begin{algorithm}
\caption{Parallelization algorithm of $F_\Delta$, with a possible reduction in size}
\begin{algorithmic}
  \Input
  \Desc{$F$}{local functions of a network over $S$, encoded as circuits}
  \Desc{$\Delta$}{update schedule over $S$}
%  \Desc{$G_U(F_\Delta)$}{update digraph of $F_\Delta$}
  \EndInput
  \Output
  \Desc{$F'$}{local functions of a parallel network over a subset of $S$,}
  \Desc{}{encoded as circuits}
  \EndOutput \parspacing
\State let $F' \gets$ apply Algorithm~\ref{algo-partiel} to $F_\Delta$%, $G_U$
\State let $G_I(F') \gets$ the interaction digraph of $F'$
\For{$(u, v) \in S^2$}
  \If{$\forall x \in \Sigma^n, f_u(x) = g(f_v(x))$} \Comment{for some $g : \Sigma \to \Sigma$}
    \For{$(u, w) \in E(G_I(F'))$}
      \State apply the substitution $x_u \mapsto g(x_v)$ in $f_w$
    \EndFor
  \EndIf
\EndFor
\While{$\exists u \in S$ such that $u$ has no accessible neighbor in $G'_I$}\
  \State $S \gets S \setminus \{u\}$ \Comment{$u$ is removed from the network}
\EndWhile
\State return $F'$
\end{algorithmic}
\label{algo-para}
\end{algorithm}

Algorithm~\ref{algo-para} is our contribution to this process,
and removes automata that are not necessary for the
limit dynamics of the network.
It proceeds in two steps:
first, the algorithm identifies pairs of automata with equivalent
local functions, up to some function. In other terms, if one automaton $u$ can be
computed as a function $g$ of the local function of another automaton $v$, then $u$
is not necessary and all references to $x_u$ in the network can be replaced
by $g(x_v)$ for an identical result. Of course, this only works under the
hypothesis that $u$ and $v$ are updated synchronously, which is the case
after the application of Algorithm~\ref{algo-partiel}.
Second, the algorithm iteratively
removes any automaton that has no influence in the network,
that is, that has no accessible
neighbor in the interaction graph of the network. These automata are not
part of cycles and do not lead to cycles, and as such
have no impact on the attractors. \parspacing

Algorithm~\ref{algo-para} is non-deterministic, and when the local functions
of any pair of automata
$(u, v)$ are shown to be equivalent up to some reversible function
$g : \Sigma \to \Sigma$, either automata could replace the influence of the other
without preference. As such, more than one result network is possible, but
all are equivalent in their limit dynamics, as will be shown later. \parspacing

While it is clear that Algorithm~\ref{algo-para} always terminates,
its complexity is out of the deterministic polynomial range, as applying
it implies
solving the \coNP-complete decision problem of testing if two
Boolean formul\ae\ are equal for all possible pairs of automata and for every
possible function $g : \Sigma \to \Sigma$.
As such, a polynomial implementation of this
algorithm would (at least) imply P $=$ \NP. This drastic conclusion is softened
when looking at restricted classes of networks where redundancies can be easily
pointed out, which is the case for the
rest of the paper.

\begin{theorem}
\label{theorem-eq-limit}
For any $F_\Delta$,
Algorithm~\ref{algo-para} returns a network $F'$ such that the limit dynamics
of $F_\Delta$ and $F'_\pi$ are isomorphic.
\end{theorem}

\begin{proof}
By Theorem~\ref{th-partiel}, the network $F'$ returned by the application of
Algorithm~\ref{algo-partiel} to $F_\Delta$ has identical dynamics
to $F_\Delta$. \smallskip

Algorithm~\ref{algo-para} operates two kinds of modifications. \smallskip

The first operation is replacing the influence of any automaton $u$
by another automaton $v$ if
they are found to have equivalent local function up to some $g:\Sigma\to\Sigma$,
that is, $f_u = g \circ f_v$.
For any configuration $x$, the value
of $f_u(x)$ and $g(f_v(x))$ are always equal. Thus, substituting the variable
$x_u$ by $g(x_v)$ in the local functions of every out-neighbor of $u$
will lead to an identical limit behavior. After this substitution,
the automaton $u$ does not have any influence over the network. Morever, all its
previous out-neighbors in $G_I(F')$ are now the out-neighbors of $v$.  \smallskip

The second operation is iteratively removing automata that do not influence
any automaton. Let
$u$ be such an deleted automaton. Consider a limit cycle $(x^1, x^2, \ldots,$ $x^k)$
in $G$. By definition of a limit cycle, $G(x^i) = x^{i + 1}$ for any $i$,
$G(x^k) = x^1$, and $x^i = x^j \implies i = j$.
Consider the component $x_u^i$ for some $i$. Since $u$ does not influence any
automaton, $x^{i + 1}$ is a function of $x^i|_{S\setminus\{u\}}$.
As the entire sequence is aperiodic, the sequence of the subconfigurations
$x^i|_{S\setminus\{u\}}$ is also aperiodic, and the attractor is preserved
in $F'$. \qed
\end{proof}

\section{Reductions in size of tangential cycles}
\label{sec-cycles}

In this section, we characterize the reduction in size that our algorithm provides
on a specific family of networks. We call
\emph{tangential cycles} (TC) any AN composed of any number of cycles
$\{C_1, C_2, \ldots, C_k\}$
%such that every cycle shares a unique path of automata, called the \emph{tangent}.
such that a unique path of automata, called the \emph{tangent},
is shared by all of the cycles.
The first automaton of the tangent is the only
automaton with more than one in-neighbor, and is called the \emph{central automaton}.
A TC is represented as part of
Figure~\ref{fig-appalgo}, which contains three cycles and a tangent of length
$0$ (only one node is shared between the cycles).
%These networks are strongly connected,
%that is, there exists a path from any node to any other in the interaction digraph
%of the network.

\subsection{Reducing block-sequential TCs}

The reduction in size provided by Algorithm~\ref{algo-para}
can be quite large on TCs, as even TCs updated in parallel can be reduced in size
by merging the different cycles as much as possible.
As such, the reduction power of this algorithm is greater than just removing
the redundancies inherent to the block-sequential to parallel update translation.
Indeed, Figure~\ref{fig-appalgo} provides an example of a parallel TC,
the size of which
is greatly reduced by the application of Algorithm~\ref{algo-para}.
But, by this process, the final result of
Algorithm~\ref{algo-para} is no longer a TC. \parspacing

As explained above, TCs
are studied as the next simplest cases of complex ANs that make biological
sense, after automata cycles.
Both isolated cycles and double cycles are examples of TCs.
To show that the study of TCs under block-sequential update schedules can be
directly reduced to the study of TCs under the parallel update schedule,
we provide an algorithm that transforms any TC under a block-sequential
update schedule into a TC under the parallel update schedule,
such that their limit dynamics are isomorphic, and the local functions of their
central automaton equivalent. This is done by simply stopping the process
of Algorithm~\ref{algo-para}
earlier to preserve the TC shape of the network.

%\vspace{-.4cm}
\begin{algorithm}
\caption{Parallelization algorithm of a TC $F$ under the block-sequential update
schedule $\Delta$, with a possible reduction in size}
\begin{algorithmic}
  \Input
  \Desc{$F$}{local functions of a network over $S$, encoded as circuits}
  \Desc{$\Delta$}{update schedule over $S$}
%  \Desc{$G_U(F_\Delta)$}{update digraph of $F_\Delta$}
  \EndInput
  \Output
  \Desc{$F'$}{local functions of a parallel network over a subset of $S$,}
  \Desc{}{encoded as circuits}
  \EndOutput \parspacing
\State let $F' \gets$ apply Algorithm~\ref{algo-partiel} to $F_\Delta$
\State let $G_I(F') \gets$ the interaction digraph of $F'$
\For{$(u, v) \in S^2$, such that either $u$ or $v$ has more than one in-neighbor in $G_I(F')$} %TODO rephrase
  \If{$\forall x \in \Sigma^n, f_u(x) = g(f_v(x))$} \Comment{for some $g : \Sigma \to \Sigma$}
    \For{$(u, w) \in E(G_I(F'))$}
      \State apply the substitution $x_u \mapsto g(x_v)$ in $f_w$
    \EndFor
  \EndIf
\EndFor
\For{$u \in S$}
  \If{$u$ has no accessible neighbors in $G_I(F')$}
    \State $S \gets S \setminus \{u\}$ \Comment{$u$ is removed from the network}
  \EndIf
\EndFor
\State return $F'$
\end{algorithmic}
\label{algo-TC}
\end{algorithm}

\begin{figure}%[t!]
\tikzset{smallstate/.style = {state, inner sep=0pt, minimum size=.8cm}}

\centering
\begin{minipage}{.3\textwidth}
\centering
\vspace{.4cm}
\begin{tikzpicture}  [node distance = .5cm, -{Latex[length=1.5mm, width=1.5mm]}]
	\node[smallstate] (A) {a};
	\node[smallstate] (B) [above right =of A] {b};
	\node[smallstate] (C) [below right =of B] {c};
	\node[smallstate] (D) [below right =of A] {d};
	\node[smallstate] (E) [above left =of A] {e};
	\node[smallstate] (F) [below left =of A] {h};

	\path
			(A) edge [bend left] node [below right] {\tiny $<$} (B)
			(B) edge [bend left] node [below left] {\tiny $\geqslant$} (C)
			(C) edge [bend left] node [above left] {\tiny $\geqslant$} (D)
			(D) edge [bend left] node [above right] {\tiny $\geqslant$} (A)
			(A) edge [bend right] node [below left] {\tiny $\geqslant$} (E)
			(E) edge node [right] {\tiny $\geqslant$} (F)
			(F) edge [bend right] node [above left] {\tiny $<$}  (A);
			
	 \draw (A) to [out=20, in=-20, looseness=7] node [right] {\tiny $\geqslant$} (A);
\end{tikzpicture}
\end{minipage}
\begin{minipage}{.3\textwidth}
\centering
 \begin{align*}
   f_a(x) &= x_a \vee x_d \vee \neg x_h \\
   f_b(x) &= \neg x_a \\
   f_c(x) &= x_b \\
   f_d(x) &= x_c \\
   f_e(x) &= x_a \\
   f_h(x) &= x_e
 \end{align*}
\end{minipage}
\begin{minipage}{.3\textwidth}
\centering
 Initial network, $\Delta = (\{h\}, \{a,c,d,e\}, \{b\})$
\end{minipage}

\vspace{.2cm}
\hrulefill
\vspace{-.3cm}

\begin{minipage}{.3\textwidth}
\hfill
\begin{center}
N / A
\end{center}
\hfill
\end{minipage}
\begin{minipage}{.3\textwidth}
\centering
 \begin{align*}
   f_a(x) &= x_a \vee x_d \vee \neg \theta_h \\
   f_b(x) &= \neg \theta_a
%   f_c(x) &= x_b \\
%   f_d(x) &= x_c \\
%   f_e(x) &= x_a \\
%   f_h(x) &= x_e
 \end{align*}$\ldots$
\end{minipage}
\begin{minipage}{.3\textwidth}
\centering
 Algorithm~\ref{algo-partiel}, first loop (c, d, e and h are unchanged)
\end{minipage}

\vspace{.2cm}
\hrulefill
\vspace{-.3cm}

\begin{minipage}{.3\textwidth}
\centering
\vspace{.5cm}
\begin{tikzpicture}  [node distance = .5cm, -{Latex[length=1.5mm, width=1.5mm]}]
	\node[smallstate] (A) {a};
	\node[smallstate] (B) [above right =of A] {b};
	\node[smallstate] (C) [below right =of B] {c};
	\node[smallstate] (D) [below right =of A] {d};
	\node[smallstate] (E) [above left =of A] {e};
	\node[smallstate] (F) [below left =of A] {h};

	\path
			(B) edge [bend left] (C)
			(C) edge [bend left] (D)
			(D) edge (A)
			(D) edge (B)
			(A) edge [-, dashed, bend left=20] (B)
			(A) edge [bend right=20] (B)
			(A) edge [bend right=20] (E)
			(E) edge [bend right=20] (A)
			(E) edge (F)
			(E) edge (B);
			
	 \draw (A) to [out=290, in=250, looseness=7] (A);
\end{tikzpicture}
\end{minipage}
\begin{minipage}{.3\textwidth}
\centering
 \begin{align*}
   f'_a(x) &= x_a \vee x_d \vee \neg x_e \\
   f'_b(x) &= \neg ( x_a \vee x_d \vee \neg x_e ) \\
   f'_c(x) &= x_b \\
   f'_d(x) &= x_c \\
   f'_e(x) &= x_a \\
   f'_h(x) &= x_e
 \end{align*}
\end{minipage}
\begin{minipage}{.3\textwidth}
\centering
 Algorithm~\ref{algo-partiel}, second loop
 %The functions $f_a$ and $f_b$ are
 %now equivalent up to $g = \neg$.
\end{minipage}

\vspace{.2cm}
\hrulefill
\vspace{.3cm}

\begin{minipage}{.3\textwidth}
\centering
\begin{tikzpicture}  [node distance = .5cm, -{Latex[length=1.5mm, width=1.5mm]}]
	\node[smallstate] (A) {a};
	\node[smallstate] (D) [below right =of A] {d};
	\node[smallstate] (C) [above right =of D] {c};
	\node[smallstate] (E) [above left =of A] {e};

	\path
			(C) edge [bend left] (D)
			(D) edge [bend left] (A)
			(A) edge [bend right] (E)
			(E) edge [bend right] (A)
			(A) edge (C)
			(E) edge [bend left, dashed, -] (C);
			
	 \draw (A) to [out=245, in=205, looseness=7] (A);
\end{tikzpicture}
\end{minipage}
\begin{minipage}{.3\textwidth}
\centering
\vspace{-.5cm}
 \begin{align*}
   f'_a(x) &= x_a \vee x_d \vee \neg x_e \\
   f'_c(x) &= \neg x_a \\
   f'_d(x) &= x_c \\
   f'_e(x) &= x_a
 \end{align*}
\end{minipage}
\begin{minipage}{.3\textwidth}
\centering
 Algorithm~\ref{algo-para} after the deletion of $h$ and the merge of $b$ into
 $a$. This is where Algorithm~\ref{algo-TC} ends
\end{minipage}

\vspace{.2cm}
\hrulefill
\vspace{.2cm}

\begin{minipage}{.3\textwidth}
\centering
\begin{tikzpicture}  [node distance = .5cm, -{Latex[length=1.5mm, width=1.5mm]}]
	\node[smallstate] (A) {a};
	\node[smallstate] (D) [below right =of A] {d};
	\node[smallstate] (C) [above right =of D] {c};

	\path
			(A) edge [bend left] (C)
			(C) edge [bend left] (D)
			(D) edge [bend left] (A)
			(C) edge [bend left] (A);
			
	 \draw (A) to [out=200, in=160, looseness=7] (A);

\end{tikzpicture}
\end{minipage}
\begin{minipage}{.3\textwidth}
\centering
\vspace{-.3cm}
 \begin{align*}
   f'_a(x) &= x_a \vee x_d \vee \neg x_c \\
   f'_c(x) &= x_a \\
   f'_d(x) &= \neg x_c
 \end{align*}
\end{minipage}
\begin{minipage}{.3\textwidth}
\centering
 End of Algorithm~\ref{algo-para}, after the merge of $e$ into
 $c$
\end{minipage}

\caption{Application of Algorithm~\ref{algo-partiel},~\ref{algo-para}
and~\ref{algo-TC}
on an example network.
Different steps of the algorithm are represented and separated using
horizontal lines. At each step, the interaction graph or update graph and
the local functions are the result of the operations indicated on the right.
As the initial network is a TC, the fourth step represents the result
returned by Algorithm~\ref{algo-TC}, which is a TC of smaller size.
The fifth step represents the result returned by Algorithm~\ref{algo-para},
which is not a TC.
Dashed lines in the interaction digraph connect automata the local function
of which are equivalent up to a negation.
Only the first graph is represented as an update digraph, as all the other
networks are updated in parallel.}
\label{fig-appalgo}
\end{figure}
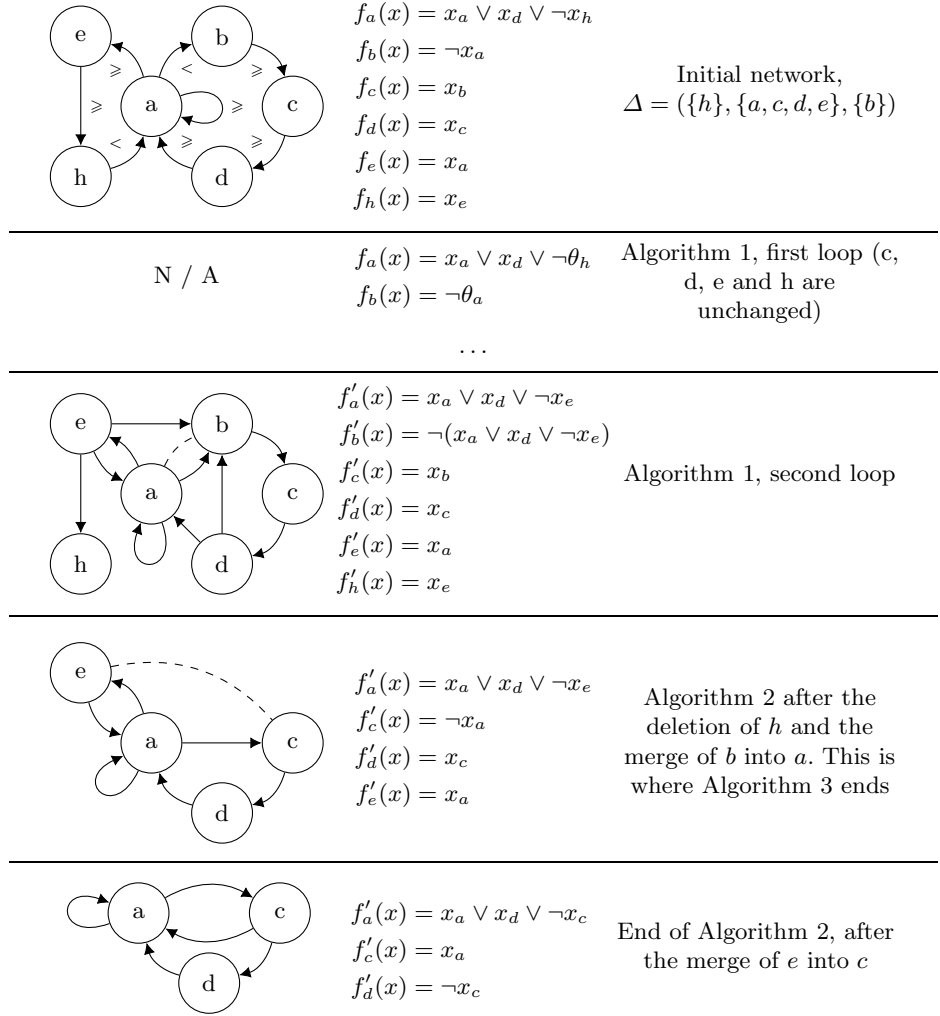
 
The only difference between Algorithms~\ref{algo-para} and~\ref{algo-TC}
is that the latter restricts the reductions it operates. If two
local functions are found to be equivalent up to some function $g$,
Algorithm~\ref{algo-TC}
removes a node if and only if these local functions are duplicates
of the previous local function of the central automaton of the network.
Removing duplications of any function that is part of a cycle would merge
two cycles and the network would no longer be tangential cycles, in a
way that is harder to count the reductions for. Since Algorithm~\ref{algo-TC}
is a variation of Algorithm~\ref{algo-para} that only does less reductions,
Theorem~\ref{theorem-eq-limit} still applies in its case. An application of
Algorithms~\ref{algo-para} and~\ref{algo-TC} is illustrated in
Figure~\ref{fig-appalgo} and the difference between the algorithms is
highlighted.

\begin{theorem}
\label{theorem-size-reduction}
Let $F$ be a TC and
$\Delta$ a block-sequential update schedule.
The amount of reductions in size that
Algorithm~\ref{algo-TC} operates on $F_\Delta$
is the number of $<$-edges in the update digraph of $F_\Delta$,
and the result is a TC.
\end{theorem}
\begin{proof}
%All possible reductions on $F$ must come from the tranformations
%applied by Algorithm~\ref{algo-partiel} to $F_\Delta$.

Algorithm~\ref{algo-partiel} operates a substitution for every $<$-edge
in the update digraph of $F_\Delta$. In this proof, we will show that
each of the possible transformations implies the removal of exactly one node
from the network. \smallskip

For any such edge $(u, <, v)$, there are two cases. Either $u$ is the central
automaton, or not. In any case, $u \neq v$ since the contrary would imply
that an automaton is updated strictly before itself. \smallskip

If we suppose that $u$ is the central automaton, this means that $f_v$ is a local
function that only depends on $x_u$. It can thus be written $f_v(x) = g(x_u)$
for some $g : \Sigma \to \Sigma$. After the application of
Algorithm~\ref{algo-partiel}, we thus obtain $f_v(x) = g(f_u(x))$, which
implies the removal of either $u$ or $v$ (but at this point, not both)
by Algorithm~\ref{algo-TC}. \smallskip

If we suppose that $u$ is not the central automaton, this means that $f_v$ is an
arbitrary formula which contains $x_u$, and $f_u$ is a function of the form
$f_u(x) = g( x_w )$ for some $g$ and some $w \in S$.
Note that $w \neq u$ by the hypothesis that $F$ is a TC, either $w$ is the previous
automaton in the path, or it is the central automaton $v$.
As such, applying Algorithm~\ref{algo-partiel} substitutes any
mention of $x_u$ in $f_v$ by $g( x_w )$.
Previously, $u$ only had one accessible neighbor, as it was part of a path
connecting to the central automaton.
This leaves $x_u$ without any accessible neighbors in the interaction digraph
of $F$, which means that it is removed by Algorithm~\ref{algo-TC}.
If the removed edge is part of a cycle, this means that this cycle will be
reduced in size. If the edge is part of the tangent, this means that the
tangent will be reduced in size. \smallskip

We thus obtain that the number of reductions is at least the number of
$<$-edges in the update digraph of the network. Suppose now that some extra
automaton $u$ is removed on top of any $<$-edge related reduction.
First observe that if $u$ has no accessible neighbor, it must have had none
from before the application of Algorithm~\ref{algo-partiel}, since in none of 
the two cases are external automaton disconnected from each other.
Now suppose that $f_u$ is equivalent to some $f_v$ up to some $g$.
Neither $u$ nor $v$ can be the central automaton, as any duplication of that
function is handled in the first case. This proves that the number of reductions
is exactly the number of $<$-edges in the update digraph of $(F, \Delta)$. \smallskip

Let us now show that the result of Algorithm~\ref{algo-TC} is a TC.
If the initial network had a central automata, there still exists a
unique central automata at the end of the algorithm, even if the original
central automata was removed in a chosen reduction.
Paths that exit the central automata
in the previous network still exit the central automata in the result,
in the same number, and still share some tangent.
The paths can be smaller in size, as well as the tangent,
but they still end in the central automata.\qed
%TODO faire une figure des trois cas

\end{proof}

If Algorithm~\ref{algo-para} cannot be polynomial
in the worst case under the hypothesis
that $\text{P} \neq \text{NP}$, Algorithm~\ref{algo-TC} can be simplified to the
following
rule: taking a TC with a block-sequential
update schedule, we obtain the equivalent parallel TC by reducing
each cycle by the number of $<$-edges that its update digraph contains.
This process is quadratic, since we only need to check the possible $<$-edges
defined by the partial order defined by $\Delta$, which are no more than $n^2$.
%This process is linear when the interaction digraph is provided as part of
%the input, since what is required is to count the number of $<$-edges
%along the various paths.

\subsection{Reducing parallel Boolean TCs further}

Applying Algorithm~\ref{algo-para} to its full extent to a Boolean 
TC (That is, a TC defined over the Boolean alphabet)
may result in
a larger reduction in size. As any automaton that is not the central one
has a unary function as its local function, any pair of non-central local functions
is equivalent up to some $g:\Sigma\to\Sigma$ if they are influenced
by the same automaton. For example, if the central automaton influences
three other automata that represent the start of three chains, these three automata
can be merged into one. Continuing this zipping process yields a final network
only as large as the longest cycle of the initial TC. %\parspacing

This process is not straightforward for non-Boolean TCs, as the local functions
along the chains can be non-reversible using modular arithmetics, for example.
Optimizing these networks is still possible, but requires a more complex set
of substitutions to do so. It has been proven in general using modules
and output functions~\cite{perrotin2021simulation}. The following theorem
corresponds to
the Boolean case, proven with more classical means.
An example of its application is
illustrated in the two last steps of Figure~\ref{fig-appalgo}.

\begin{theorem}
Let $F$ be a Boolean TC. Applying Algorithm~\ref{algo-para} to $F_\pi$
generates a network $F'$ whose size is that of the largest cycle in $F$.
\label{theorem-bool-reduction}
\end{theorem}
\begin{proof}
Starting from the initial TC $F$, all of the automata directly influenced by the
automata at the end of the tangent $u$ (but that are not $u$)
have local functions
$f_v(x) = g(x_u), f_w(x) = h(x_u), \ldots$ for $g, h, \ldots : \Sigma = \{0,1\}\to\Sigma$.
All these functions $g, h, \ldots$ are not constant functions, since the automata
that they represent are influenced by an automaton by hypothesis. Thus, they can
only be defined as the identity or the negation of $x_u$. As
a consequence, all but one of these automata
will be removed by the algorithm as they are all equivalent up to some $g$. \smallskip

This same argument can be repeated by taking all the automata influenced by the
only automaton resulting from the previous iteration, excluding the central
automaton. At each step, all of the automata at the same distance from the
central automaton are merged. Hence, at the end of this process,
whatever the choices made for merging automata along the iterative process,
the resulting AN will be compoesed of $k$ automata, with $k$ the length of
the largest cycle of $F$.
\qed
\end{proof}

\section{An application: disjunctive double cycles}
\label{sec-dc}

As an application of this algorithm and as an example to the algorithm's
capacities to reduce the size of the provided network, we turn to the family
of disjunctive double cycles. Notice that the result still holds for conjunctive
double cycles since conjuctive and disjunctive cycles have isomorphic
dynamics~\cite{T-Noual2012,n12b}.
\parspacing

In disjunctive automata networks, an edge $(u, v)$ is signed positively if the
$x_u$ appears as a positive variable in $f_v$.
An edge $(u, v)$ is signed negatively if $x_u$ appears as a negative
variable in $f_v$.
A cycle is said to be positive if it contains an even number of negative
edges, and negative otherwise. \parspacing

A \emph{disjunctive double cycle} is an automata network with an interaction
digraph that is composed of two automata cycles that intersect in one
automaton. The local function of this central automaton is a disjunctive clause.
This family of networks is very simple to define, and is a simple and intuitive
next step after the family of Boolean automata cycles, which are composed of a single
cycle. \parspacing

Both families have been characterized under the parallel update schedule~\cite{T-Noual2012,J-Demongeot2012};
that is to say, given basic parameters concerning the size of the
cycles, their sign, and any integer $k$,
an explicit formula (defined as a polynomially computable function) has
been given among other to count the number of limit cycles of size $k$ of such
networks under the parallel update schedule.
In this section, we extend this characterization to
the block-sequential equivalents
by showing how applying our algorithm reduces the
network to a smaller instance of the same family of networks. \parspacing

Furthermore, as Boolean automata cycles and disjunctive double
cycles are TCs, our method can be simplified to the simple following rules: 
given a TC $F$,
a block-sequential update schedule $\Delta$, count the
number of $<$-edges in the update digraph $G_U(F_\Delta)$;
for every cycle, substract
to its size the number of such edges it contains, while keeping its sign;
the final network, under the parallel update schedule, and the initial
network under $\Delta$ have isomorphic limit dynamics. This is a simple
application of Theorem~\ref{theorem-size-reduction}, and of the rule of
thumb deduced from Algorithm~\ref{algo-TC}. \parspacing

We denote by $DC(s, s', a, b)$ the disjunctive double cycles with cycle sizes
$a, b$ and signs $s, s'$.

\begin{theorem}
Let $D = DC(s, s', a, b)$ be disjunctive double cycles, $\Delta$ a
block-sequential update schedule. For $A$ ($B$ respectively) the number of $<$-edges on the
cycle of size $a$ ($b$ respectively) in $G_U(F_\Delta)$, the limit dynamics of
$D_\Delta$ is isomorphic to that of
$D'_\pi$, where $D' = DC(s, s', a - A, b - B)$.
\end{theorem}
\begin{proof}
This is a straightforward application of Theorem~\ref{theorem-size-reduction}.\qed
\end{proof}

\section{Conclusion}

In this paper we provide a novel algorithm which allows the reduction in size
of automata networks, in particular by passing the network from a
block-sequential to a parallel update schedule, while keeping isomorphic
limit dynamics. While this algorithm is too
computationally expensive for the general case, we study the specific family
of intersection of automata cycles, on which this algorithm is
easily applied. This study allows the discovery that all block-sequential
tangential cycles have isomorphic limit dynamics to parallel
tangential cycles. Finally, we apply this fact to Boolean automata double cycles
to characterize their behavior under block-sequential update schedules. \parspacing

It seems now clear to us that the difference between the parallel update
schedule and block-sequential update schedules is that the latter changes
the timing of the information along sections of the network. In particular,
structures such as tangential cycles can be directly
translated into an equivalent parallel network with shorter cycles.
We are interested in seeing what effects this translation could have in
a more general set of families of networks, and if there exists other families
in which block-sequential update schedules lead to equivalent parallel networks
which are still part of the family. \parspacing

As a perspective, we would like to characterize more redundancies that
can be removed from networks to help with the computation of their dynamics.
For example, we are currently interested in more complex compositions
of automata cycles, and have already found equivalences that show that many
networks are equivalent in their limit dynamics where complex parts of
automata networks can be moved alongside cycles without affecting the network's
limit dynamics. \parspacing

Isolated paths are also a strong candidate for size reduction. Isolated
paths are paths that lead from cycles to other cycles but can only
be crossed once. Our current
algorithms conserve such paths, despite it being possible to reduce them
completely without changing the limit dynamics of the network in many cases,
for example when an isolated path is the only way to go from one part to
another.
We have to be careful when multiple isolated paths exit from and join onto
the same parts, as the synchronicity of the information in the entire network
must be preserved.

\bigskip
\noindent
\small
\textbf{Acknowledgments.}
This work has been partially funded by ANR-18-CE40-0002 FANs project
(PP \& SS), ECOS-Sud CE19E02 SyDySy project (PP \& SS), and STIC-AmSud
22-STIC-02 CAMA project (SS).

\bibliographystyle{plain}
{\small{\bibliography{bloc}}}

%\appendix
%\include{proofs}

\end{document}